\documentclass[a4paper,11pt]{amsart}

\usepackage[english]{babel}
\usepackage[T1]{fontenc}
\usepackage[utf8x]{inputenc}
\usepackage{xspace}
\usepackage{hyperref}
\usepackage{subfigure}
\usepackage{color}
\usepackage{mathrsfs}
\usepackage{framed}

\usepackage{mathrsfs}
\usepackage{amsfonts}
\usepackage{amssymb}
\usepackage{amsmath}
\usepackage{array,multirow}
\usepackage{mathbbol}

\usepackage{color}
\usepackage{graphicx}
\usepackage{graphics}

\usepackage[lined, ruled,vlined,algo2e]{algorithm2e}
\usepackage{mathabx}

\usepackage{lineno}

\def\cO{\mathcal{O}}

\def\cC{\mathcal{C}}

\def\cH{\mathcal{H}}
\def\cE{\mathcal{E}}

\def\cT{\mathcal{T}}

\def\cF{\mathcal{F}}
\def\cX{\mathcal{X}}
\def\cZ{\mathcal{Z}}

\newcommand{\inc}[1]{\ensuremath{\widetilde{N}(#1)}\xspace}

\newcommand*\patchAmsMathEnvironmentForLineno[1]{%
  \expandafter\let\csname old#1\expandafter\endcsname\csname #1\endcsname
  \expandafter\let\csname oldend#1\expandafter\endcsname\csname end#1\endcsname
  \renewenvironment{#1}%
     {\linenomath\csname old#1\endcsname}%
     {\csname oldend#1\endcsname\endlinenomath}}%
\newcommand*\patchBothAmsMathEnvironmentsForLineno[1]{%
  \patchAmsMathEnvironmentForLineno{#1}%
  \patchAmsMathEnvironmentForLineno{#1*}}%
\AtBeginDocument{%
\patchBothAmsMathEnvironmentsForLineno{equation}%
\patchBothAmsMathEnvironmentsForLineno{align}%
\patchBothAmsMathEnvironmentsForLineno{flalign}%
\patchBothAmsMathEnvironmentsForLineno{alignat}%
\patchBothAmsMathEnvironmentsForLineno{gather}%
\patchBothAmsMathEnvironmentsForLineno{multline}%
}

\def\dom{{\sc Dom-Enum}\xspace}
\def\transhyp{{\sc Trans-Enum}\xspace}

\def\np{{\sc NP}\xspace}

\newcommand{\eg}{\emph{e.g.}\xspace}
\newcommand{\ie}{\emph{i.e.}\xspace}

\newcommand{\macro}[3]{\newcommand{#1}[#3]{#2}}
\macro{\size}{\|#1\|}{1}

\newtheorem{thm}{Theorem}

\newtheorem{prop}[thm]{Proposition}

\newtheorem{theorem}[thm]{Theorem}
\newtheorem{lemma}[thm]{Lemma}
\newtheorem{proposition}[thm]{Proposition}
\newtheorem{definition}[thm]{Definition}

\title{Polynomial Delay Algorithm for Listing Minimal Edge Dominating sets in Graphs}
\author[M.M. Kant{\'e} \and V. Limouzy\and A. Mary \and L. Nourine \and T. Uno]{Mamadou Moustapha Kant{\'e} \and Vincent Limouzy \and Arnaud
  Mary \and Lhouari Nourine \and Takeaki Uno}

\address{Clermont-Universit{\'e}, Universit{\'e} Blaise Pascal, LIMOS,
  CNRS, France} 
\email{\{mamadou.kante,limouzy,mary,nourine\}@isima.fr}
\address{National Institute of Informatics, Japan}
\email{uno@nii.jp}
\thanks{M.M. Kant{\'e} and V. Limouzy are supported by the French Agency
  for Research under the DORSO project.}

\begin{document}

\begin{abstract}
  A \emph{hypergraph} is a pair $(V,\cE)$ where $V$ is a finite set and $\cE\subseteq 2^V$ is called the set of hyper-edges. An output-polynomial algorithm for $\cC\subseteq 2^V$ is an algorithm that
  lists without repetitions all the elements of $\cC$ in time polynomial in the sum of the size of $\cH$ and the accumulated size of all the elements in $\cC$. Whether there exists an
  output-polynomial algorithm to list all the inclusion-wise minimal hitting sets of hyper-edges of a given hypergraph (the \transhyp problem) is a fifty years old open problem, and up to now
  there are few tractable examples of hypergraph classes.  An inclusion-wise minimal hitting set of the closed neighborhoods of a graph is called a minimal dominating set. A closed neighborhood of a
  vertex is the set composed of the vertex itself with all its neighbors. It is known that there exists an output-polynomial algorithm for the set of minimal dominating sets in graphs if and only if
  there is one for the minimal hitting sets in hypergraphs.  Hoping this equivalence can help to get new insights in the \transhyp problem, it is natural to look at graph classes.  It was proved
  independently and with different techniques in [Golovach et al. - ICALP 2013] and [Kant\'e et al. - ISAAC 2012] that there exists an incremental output-polynomial algorithm for the set of minimal
  edge dominating sets in graphs (\ie minimal dominating sets in line graphs).  We provide the first polynomial delay and polynomial space algorithm that lists all the minimal edge dominating sets in
  graphs, answering an open problem of [Golovach et al. - ICALP 2013]. Besides the result, we hope the used techniques that are a mix of a modification of the well-known Berge's algorithm and a strong
  use of the structure of line graphs, are of great interest and could be used to get new output-polynomial algorithms.
\end{abstract}

\maketitle

\section{Introduction}\label{sec:1}

The \textsc{Minimum Dominating Set} problem is a classic and well-studied graph optimization problem. A \emph{dominating set} in a graph $G$ is a subset $D$ of its set of vertices such that each
vertex is either in $D$ or has a neighbor in $D$. Computing a minimum dominating set has numerous applications in many areas, \eg, networks, graph theory (see for instance the book \cite{HHS98}). The
\textsc{Minimum Edge Dominating Set} problem is a classic well-studied variant of the \textsc{Minimum Dominating Set} problem \cite{HHS98}. An edge dominating set is a subset $F$ of the edge set such
that each edge is in $F$ or is adjacent to an edge in $F$. In this paper, we are interested in an output-polynomial algorithm for listing without duplications the (inclusion-wise) minimal edge
dominating sets of a graph. An output-polynomial algorithm is an algorithm whose running time is bounded by a polynomial depending on the sum of the sizes of the input and output. The enumeration of
minimal or maximal subsets of vertices satisfying some property in a (hyper)graph is a central area in graph algorithms and for several properties output-polynomial algorithms have been proposed
\eg \cite{AvisF96,BorosEG06,EGM03,FominHKPV14,lawler80,SchwikowskiS02,Tarjan73}, while for others it was proven that no output-polynomial time algorithm exists unless P=NP
\cite{KhachiyanBBEG08b,KhachiyanBBEGM08a,KhachiyanBEG08,lawler80,Strozecki10}.

The existence of an output-polynomial algorithm for the enumeration of minimal dominating sets of graphs (\dom problem) is a widely open question and is closely related to the well-known \transhyp
problem in hypergraphs which asks for an output-polynomial algorithm for the enumeration of all minimal transversals in hypergraphs. A \emph{transversal} (or a \emph{hitting-set}) in a hypergraph is a
subset of its vertex set that intersects every of its hyper-edges. This is a long-standing open problem (see for instance \cite{EG95}) and is well-studied due to its applications in several areas
\cite{EG95,EGM03,GunopulosKMT97,BEKG06,nourine12}. Up to now only few tractable cases are known (see \cite{KLMN11} for some examples). It is easy to see that the minimal dominating sets of a graph are
the minimal transversals of its \emph{closed neighbourhoods}\footnote{The \emph{closed neighborhood} of a vertex $v$ in a graph is the set containing $v$ and all its neighbors.}, and then, as a
particular case, it seems that the \dom problem is more tractable than the \transhyp problem, but some of the authors have very recently proved in \cite{KLMN13} that there exists an output-polynomial
algorithm for \dom if and only if there exists one for \transhyp. It is therefore interesting to look at the \transhyp problem in the perspective of graph theory and to investigate graph classes where
the \dom problem is tractable. Indeed, the authors have exhibited several new tractable cases split graphs \cite{KLMN11}, undirected path graphs \cite{KLMN12}, permutation graphs \cite{KanteLMNU13}
and chordal $P6$-free graphs \cite{KLMN11}.  It is proved in \cite{KLMN12} (and independently in \cite{GolovachHKV13}) that the enumeration of minimal edge dominating sets can be done in (incremental)
output-polynomial. In this paper we improve drastically these two algorithms and give the first algorithm with polynomial delay (the time between two consecutive outputs is bounded by a polynomial on
the input size) and that uses polynomial space.

The strategy of our algorithm is based on a modified version of Berge's algorithm \cite{Berge1989} for enumerating minimal transversals. Berge's algorithm consists in ordering the hyper-edges
$E_1,\ldots,E_m$ of a given hypergraph $\cH$ and for each $i$ computes the minimal transversals of the hypergraph $\cH_i$ with set of hyper-edges $\{E_1,\ldots,E_i\}$ from the minimal transversals of
$\cH_{i-1}$.  The main drawback of Berge's algorithm is that a minimal transversal of $\cH_i$ is not necessarily a subset of a minimal transversal of $\cH$, and therefore Berge's algorithm does not
guarantee an output-polynomial algorithm.  Indeed, it is proved in \cite{Takata07} that there exist hypergraphs for which Berge's algorithm is not output-polynomial for any ordering.  We will propose
a new version of Berge's algorithm that consists in grouping the hyper-edges of a given hypergraph with respect to the hyper-edges of a \emph{maximal matching}\footnote{A \emph{maximal matching} in a
  hypergraph is a set of disjoint edges maximal with respect to inclusion.}. This overcomes the inconvenient of Berge's algorithm, but the counterpart is the difficulty in the extension of the already
computed sub-transversals (which was trivial in Berge's algorithm). We will show how to solve this extension part in the case of minimal edge dominating sets.  It is worth noticing that a strategy
similar to the Berge's algorithm exists for enumerating maximal subsets satisfying a certain property \cite{lawler80}. It consists in ordering the vertices $x_1,\ldots,x_n$ of a hypergraph and for
each $i$ computes the maximal sets that are maximal within $\{1,\ldots,i\}$. But contrary to Berge's algorithm every sub-solution computed at step $i$ is a subset of a solution, and therefore as in
our version of Berge's algorithm if one can compute the maximal sets that are maximal within $\{1,\ldots,i\}$ from those that are maximal within $\{1,\ldots,i-1\}$ one gets an output-polynomial
algorithm.  This strategy has been used by Lawler et al. in \cite{lawler80} to give output-polynomial algorithms for enumerating maximal set packings, maximal complete k-partite subgraphs, the bases
of the intersection of $m$ matroids for fixed $m$, \ldots



The remainder of the text is as follows.  In Section \ref{sec:2} we give some necessary definitions.  Section \ref{sec:3} explains the strategy for the enumeration and give an informal description of
the algorithm.  Section \ref{sec:4} is devoted to the technical parts of the paper.  We sum up all the details of the algorithm in Section \ref{sec:6}, and conclude with some open questions in Section
\ref{sec:7}.

\section{Preliminaries}\label{sec:2}

If $A$ and $B$ are two sets, $A\backslash B$ denotes the set $\{x\in A\mid x\notin B\}$.  The power-set of a set $V$ is denoted by $2^V$.  For three sets $A,B$ and $C$, we write $A=B\sqcup C$ if
$A=B\cup C$ and $B\cap C=\emptyset$.  The size of a set $A$ is denoted by $|A|$.  We refer to \cite{BondyM08} for graph terminology.  A graph $G=(V,E)$ is a pair of vertex set $V$ and edge set
$E\subseteq V\times V$.  We only deal with finite and simple graphs.  An edge between $x$ and $y$ in a graph is denoted by $xy$ (equivalently $yx$) and sometimes it will be convenient to see an edge
$xy$ as the set $\{x,y\}$, but this will be clear from the context.

Let $G:=(V, E)$ be a graph.  For a vertex $x$, we denote by $\inc{x}$ the set of edges incident to $x$, and $N(x)$ denotes the set of vertices adjacent to $x$.  For every edge $e:=xy$, we denote by
$N[e]$ the set of edges adjacent to $e$, \ie $N[e]:=\inc{x}\cup \inc{y}$.  A subset $D$ of $E$ is called an \emph{edge dominating set} if for every edge $e$ of $G$, we have $N[e]\cap D\ne \emptyset$,
and $D$ is \emph{minimal} if no proper subset of it is an edge dominating set.

If $V$ is a finite set and $\cC\subseteq 2^V$, then the size of $\cC$, denoted by $\size{\cC}$ is $\sum_{C\in \cC} |C|$. A \emph{hypergraph} $\mathcal{H}$ is a pair $(V,\cF)$ with a vertex set $V$ and
a hyper-edge set $\cF\subseteq 2^V$.  It is worth noticing that graphs are special cases of hypergraphs. The size of a hypergraph $\cH$, denoted by $\size{\cH}$, is defined as $|V|+\size{\cF}$.

Let $\cH$ be a hypergraph and let $\cC$ be a subset of $2^{V}$. An \emph{output-polynomial} algorithm for $\cC$ is an algorithm that lists the elements of $\cC$ without repetitions in time
$\cO\left(p\left(\size{\cH}, \size{\cC}\right)\right)$ for some polynomial $p$.  We say that an algorithm enumerates $\cC$ with \emph{polynomial delay} if, after a pre-processing that runs in time
$\cO(p(\size{\cH}))$ for some polynomial $p$, the algorithm outputs the elements of $\cC$ without repetitions, the delay between two consecutive outputs being bounded by $\cO(q(\size{\cH}))$ for some
polynomial $q$ (we also require that the time between the last output and the termination of the algorithm is bounded by $\cO(q(\size{\cH}))$). It is worth noticing that an algorithm which enumerates a subset
$\cC$ of $2^{V}$ in polynomial delay outputs the set $\cC$ in time $\cO\left(p(\size{\cH}) + q(\size{\cH})\cdot |\cC| + \size{\cC}\right)$ where $p$ and $q$ are respectively the polynomials bounding
the pre-processing time and the delay between two consecutive outputs.  Notice that any polynomial delay algorithm is obviously an output-polynomial one, but not all output-polynomial algorithms are
polynomial delay \cite{Strozecki10}. We say that an output-polynomial algorithm uses polynomial space if there exists a polynomial $q$ such that the space used by the algorithm is bounded by
$q(\size{G})$. The output polynomiality is one of a criteria of efficiency for enumeration algorithms, and polynomial delay is that of more efficiency.  This paper addresses the following question.


\ \\
\noindent
{\bf Question}: Is there a polynomial delay polynomial space algorithm to
 enumerate all minimal edge dominating sets of a given graph?
\ \\

Let $\cH:=(V,\cF)$ be a hypergraph. A \emph{private neighbor of a vertex $x$ with respect to $T\subseteq V$} is a hyper-edge that intersect $T$ only on $x$, and the set of private neighbors of $x$ is denoted by $P_\cH(x,T)$, \ie $P_\cH(x,T):=\{F\in \cF\mid F\cap T= \{x\}\}$.  A subset $T$ of $V(\cH)$ is called an \emph{irredundant set} if $P_\cH(x,T) \ne \emptyset$ for all $x\in T$.  We denote by $IR(\cH)$ the set of irredundant sets of $\cH$.

A \emph{transversal} (or \emph{hitting set}) of $\mathcal{H}$ is a subset of $V$ that has a non-empty intersection with every hyper-edge of $\cF$; it is \emph{minimal} if it does not contain any other
transversal as a proper subset.  It is known that $T$ is a minimal transversal if and only if $T$ is a transversal and $P_\cH(x,T) \ne \emptyset$ for all $x\in T$.  The set of all minimal transversals
of $\mathcal{H}$ is denoted by $tr(\mathcal{H})$.  The \emph{edge neighborhood hypergraph} $\cH(G)$ of $G$ is the hypergraph $(E, \{N[e]\mid e\in E\})$.  The following proposition is easy to obtain.

\begin{proposition}\label{prop:ednt}
For any graph $G:=(V, E)$, $T\subseteq E$ is an edge dominating set of $G$
 if and only if $T$ is a transversal of $\cH(G)$.
Therefore, $T\subseteq E(G)$ is a minimal edge dominating set of $G$
 if and only if $T$ is a minimal transversal of $\cH(G)$.
\end{proposition}

Our algorithm is based on algorithms developed for enumerating minimal transversals.  Thus, we introduce the dominating set version of the above notions.  For an edge set $T\subseteq E$, a
\emph{private neighbor of an edge $e$ with respect to $T$} is an edge $f$ that is adjacent to $e$ but not to any edge in $T\setminus \{ e\}$.  The set of private neighbors of $e$ is denoted by $P_E(e,
T)$.  For an edge subset $E'$ of $E$, an edge set $T$ is called \emph{transversal} (or \emph{hitting set}) of $E'$, if for any edge $e$ of $E'$, $T$ includes at least one edge $f$ of $E$ that is
adjacent to the edge $e$.  A transversal of $E'$ is \emph{minimal} if it does not contain any other transversal of $E'$ as a proper subset.  Again $T\subseteq E$ is a minimal transversal of $E'$ if
and only if $T$ is a transversal of $E'$ and $P_{E'}(e,T) \ne \emptyset$ for all $e\in T$.  The set of all minimal transversals of $E'$ is also denoted by $tr(E')$.

\section{Berge's Algorithm and Basic Strategy}\label{sec:3}

Hereafter, we explain our approach to enumerate all minimal edge dominating sets.  Our strategy for the enumeration is based on Berge's algorithm \cite{Berge1989}.  For a given hypergraph with
hyper-edges enumerated as $F_1,\ldots, F_m$, let $\cF_j$ be $\{F_1,\ldots, F_j\}$ for each $1\leq j \leq m$.  Roughly, Berge's algorithm computes, for each $1< j\leq m$, $tr(\cF_j)$ from
$tr(\cF_{j-1})$.  Although the algorithm is not polynomial space, there is a way to reduce the space complexity to polynomial.  For $j\geq 1$ and $T\in tr(\cF_j)$, we define the parent $Q'(T,j)$ of $T$ as
follows
\begin{align*}
  Q'(T,j)&:= \begin{cases} T & \textrm{if $T\in tr(\cF_{j-1})$},\\ T\setminus \{v\} & \textrm{if $v$ is such that $P_{\cF_j}(v,T)=\{F_j\}$}. \end{cases}
\end{align*}

We can observe that $T\not\in tr(\cF_{j-1})$ if and only if $P_{\cF_j}(v,T) = \{F_j\}$ holds for some $v\in T$, thus the parent is well defined and always in $tr(\cF_{j-1})$
\cite{KavvadiasS05,MurakamiU14}.  One can moreover compute the parent of any $T\in tr(\cF_j)$ in time polynomial in $\size{\cH}$. This parent-child relation induces a tree, rooted at $\emptyset$,
spanning all members of $\bigcup\limits_{1\leq j\leq m}tr(\cF_j)$.  We can traverse this tree in a depth-first search manner from the root by recursively generating the children of the current
visiting minimal hitting set.  Any child is obtained by adding at most one vertex, then the children can be listed in polynomial time.  In this way, we can enumerate all the minimal hitting sets of a
hypergraph with polynomial space.

Formally and generally, we consider the problem of enumerating all elements of a set $\cZ$ that is a subset of an implicitly given set $\cX$.  Assume that we have a polynomial time computable parent
function $P:\cX\rightarrow\cX\cup\{ nil\}$.  For each $X\in \cX$, $P(X)$ is called the {\em parent} of $X$, and the elements $Y$ such that $P(Y)=X$ are called {\em children} of $X$.  The parent-child
relation of $P$ is {\em acyclic} if any $X\in \cX$ is not a proper ancestor of itself, that is, it always holds that $X\ne P(P(\cdots P(X))\cdots)$.  We say that an acyclic parent-child relation is
{\em irredundant} when any $X\in \cX$ has a descendant in $\cZ$, in the parent-child relation.  The following statements are well-known in the literature
\cite{AvisF96,ShiouraTU97,MakinoU04,KavvadiasS05,MurakamiU14}.

\begin{proposition}\label{prop:poly-space} All elements in $\cZ$ can be enumerated with polynomial space if there is an acyclic parent-child relation $P:\cX\to \cX\cup \{nil\}$ such that there is a 
  polynomial space algorithm for enumerating all the children of each $X\in \cX\cup \{nil\}$.
\end{proposition}

\begin{proposition}\label{prop:poly-delai} All elements in $\cZ$ can be enumerated with polynomial delay and polynomial space if there is an irredundant parent-child relation $P:\cX\to \cX\cup
  \{nil\}$ such that there is a polynomial delay polynomial space algorithm for enumerating all the children of each $X\in \cX\cup \{nil\}$.
\end{proposition}

With acyclic (resp., irredundant) parent-child relation $P:\cX\to \cX\cup \{nil\}$, the following algorithm enumerates all elements in $\cZ$, with polynomial space (resp., with polynomial delay and
polynomial space).

\begin{tabbing}
{\bf Algorithm} {\sf ReverseSearch}$(X)$\\
1. {\bf if} $X\in \cZ$ {\bf then output} $X$\\
2. {\bf for} each $Y$ satisfying $X=P(Y)$ {\bf do}\\
3. \ \ \ {\bf call} {\sf ReverseSearch}$(Y)$\\
4. {\bf end for}
\end{tabbing}

The call {\sf ReverseSearch}$(nil)$ enumerates all elements in $\cZ$.  Since the above parent-child relation for transversals $Q'$ is acyclic, the algorithms proposed in
\cite{KavvadiasS05,MurakamiU14} use polynomial space.  However, the parent-child relation $Q'$ is not irredundant and hence {\sf ReverseSearch$(nil)$} does not guarantee a polynomial delay neither an
output-polynomiality.  Indeed, we can expect that the size of $tr(\cF_j)$ increases as the increase of $j$, and it can be observed in practice.  However, $tr(\cF_j)$ can be exponentially larger than
$tr(\cF_m)$, thus Berge's algorithm is not output polynomial \cite{Takata07}.  Examples of irredundant parent-child relations can be found in the literature \cite{AvisF96,ShiouraTU97,MakinoU04}.  

One idea to avoid the lack of irredundancy is to certificate the existence of minimal transversals in the descendants.  Suppose that we choose some levels $1=l_1,\ldots,l_k=m$ of Berge's algorithm,
and state that for any $T\in tr(\cF_{l_j})$, we have at least one descendant in $tr(\cF_{l_{j+1}})$.  This implies that any transversal in $tr(\cF_{l_j})$ has a descendant in $tr(\cF_m)$, thus we can
have an irredundant parent-child relation by looking only at these levels, and the enumeration can be polynomial delay and polynomial space.


We will use this idea to obtain a polynomial delay polynomial space algorithm to enumerate the minimal edge dominating sets, the levels are determined with respect to a \emph{maximal matching}.  From
now one we assume that we have a fixed graph $G:=(V,E)$ and we will show how to enumerate all its minimal edge dominating sets. A subset of $E$ is a \emph{matching} if every two of its edges $e$ and
$f$ are not adjacent.  A matching is \emph{maximal} if it is not included in any other matching.  Let $\{b_1,\ldots,b_k\}$ be a maximal matching of $G$, and let $b_i = x_iy_i$.  For each $0\leq i \leq
k$, let $V_i:=V\setminus \left( \bigcup\limits_{i'>i} b_{i'} \right)$, and let $E_i := \{ e \mid e\subseteq V_i\})$.  Let $B_i := E_i\setminus E_{i-1}$ for $i>1$.  Note that any edge in $E_1$ is
adjacent to $b_1$ and by definition $B_i$ never includes an edge $b_j\ne b_i$.  Without loss of generality, we here assume that we have taken a linear ordering $\leq$ on the edges of $G$ so that:~ (1)
for each $e\in B_i$ and each $f\in E_{i-1}$ we have $e< f$,~ (2) for each $e\in \widetilde{N}(x_i)\cap B_i$, each $f\in \widetilde{N}(y_i) \cap B_i$ we have $b_i < e < f$. Observe that with that
ordering we have $e<f$ whenever $e\in B_i$ and $f\in B_j$ with $i<j$. We consider that Berge's algorithm on $\cH(G)$ follows that ordering. In fact we will prove using Berge's algorithm that we can
define an irredundant parent-child relation to enumerate $tr(E_i)$ from $tr(E_{i-1})$.



\begin{lemma}\label{lem:irredundant} Let $1\leq i < k$. Any $T\in tr(E_{i-1})$ has at least one descendant in $tr(E_i)$.
\end{lemma}

\begin{proof} If $T'\in tr(E_i)$ satisfies $T' = T$, then $T'$ is a descendant of $T$ since the parent is never greater than the child.  If $T\not\in tr(E_i)$, some edges $X$ of $B_i$ are not
  dominated by $T$, and consider $T' := T\cup \{b_i\}$.  We observe that $b_i$ is adjacent to all edges of $B_i$ and the edges in $X$ are private neighbors of $b_i$ in $T'$, thus $T'$ is included in
  $tr(E_i)$.  Let us compute the ancestor of $T'$ in $E_{i-1}$ as follows: set $T":=T'$ and repeatedly compute the parent of $T"$ and set $T"$ to its parent, until reaching a minimal transversal in
  $tr(E_{i-1})$. In this process no vertex of $T$ is removed since each vertex in $T$ has a private neighbor in $E_{i-1}$. But, at some point $b_i$ is removed from $T'$ since it is the only one in
  $T'$ which has a private neighbor in $B_i$. This means that $T$ is an ancestor of $T'$, and thus $T$ always has a descendant in $tr(E_i)$.
\end{proof}

For conciseness, we introduce a new parent-child relation for edge dominating set enumeration.  For $T\in tr(E_i)$, let $Q'_j(T, |E_i|)$ be the ancestor of $T$ located on the $j$-th level of Berge's
algorithm, i.e., $Q'_j(T, |E_i|) = Q'(Q'(\cdots(T, |E_i|), |E_i|-1),\cdots, j+1)$.  Then, we define the \emph{skip parent} $Q(T, i)$ of $T$ by $Q'_{|E_i-1|}(T, |E_i|)$.  $T'$ is a \emph{skip-child} of
$T\in tr(E_{i-1})$ if and only if $T'\in tr(E_i)$ and $Q(T', i) = T$.  The set of skip-children of $T\in tr(E_i)$ is denoted by $\cC(T, i)$. From Propositions \ref{prop:poly-space} and
\ref{prop:poly-delai}, and Lemma \ref{lem:irredundant} we have the following proposition.

\begin{proposition}\label{prop:strategy} If we can list all skip-children of $T\in tr(E_i)$, for each $1\leq i \leq k$, with polynomial delay and polynomial space, then we can enumerate all minimal
  edge dominating sets with polynomial delay and polynomial space.
\end{proposition}

But, as we will show in the next section, for a transversal $T$ in $tr(E_{i-1})$, the problem of finding a transversal of $tr(E_i)$ including $T$ is \np\!\!-complete in general. In order to overcome this difficulty, we
will identify a pattern, that we call an \emph{$H$-pattern}, that makes the problem difficult.  We will first show that one can enumerate with polynomial delay and polynomial space all the
skip-children that include no edges from $H$-patterns, and then define a new parent-child relation that will allow to enumerate also with polynomial delay and polynomial space the other skip-children
in a different way.  In the following sections, we explain the methods for the enumeration.



\section{Computing Skip-Children}\label{sec:4}

Let $T$ be in $tr(E_{i-1})$ and $T'\in tr(E_i)$ a skip-child of $T$.  First notice that every edge in $T'\setminus T$ can have a private neighbor only in $B_i$. Indeed every edge in $E_{i-1}$ is
already dominated by $T$ and an edge in $T'\setminus T$ is only used to dominate an edge in $B_i$. Moreover, an edge $e\ne b_i$ in $\inc{x_i}\cap (T'\setminus T)$ (resp. in $\inc{y_i}\cap (T'\setminus
T)$) can have private neighbors only in $\inc{x_i}\cap B_i$ (resp. $\inc{x_i}\cap B_i$). And from the proof of Lemma \ref{lem:irredundant} if $b_i\in T'\setminus T$ then $T'\setminus T=\{b_i\}$. 

Let us first consider the case that every edge in $T'\setminus T$ is adjacent to $b_i$.  From our discussion above, when two edges in $T'\setminus T$ are incident to $x_i$ (resp. $y_i$), they cannot
have both private neighbors.  Thus $T'\setminus T$ can include at most two such edges.  Therefore, by choosing all combinations of one or two edges adjacent to $b_i$, adding them to $T$ and then
checking if the skip-parent of the resulting set is $T$, we can enumerate all the skip-children $T'$ of $T$ such that $T'\setminus T\subseteq B_i$ with polynomial delay and polynomial space.  

We now consider the remaining case that an edge in $T'\setminus T$ is not adjacent to $b_i$.  We call such a skip-child \emph{extra}. We can see that at least one edge $f\ne b_i$ adjacent to $b_i$ has
to be included in $T'$ to dominate $b_i$.  Actually, since $b_i< e$ for any $e\in B_i\setminus \{b_i\}$, any extra skip-child of $T$ is a descendant of some $T\cup \{f\}$ with $f\ne b_i$ incident to
$x_i$ or $y_i$ in the original parent-child relation.  So, without loss of generality, we will assume that such an edge $f\ne b_i$ is incident to $x_i$ and is included in $T$.  Hereafter, we suppose
that $N(y_i) := \{z_1,\ldots,z_k\}$ and assume $T'$ is an extra skip-child of $T$.

A vertex $z_h\in N(y_i)\cap V_i$ is \emph{free} if it is not incident to an edge in $T$, and is \emph{non-free} otherwise.  A free vertex is said to be \emph{isolated} if it is not incident to an edge
in $E_{i-1}$.  Clearly, if there is an isolated free vertex, then $T$ has no extra skip-child.  Thus, we assume that there is no isolated free vertex.  Edges in $E_i\setminus B_i$ that are incident to
some free vertices are called \emph{border edges}.  Observe that any border edge $vz_h$ incident to a free vertex $z_h$ is adjacent to an edge $vw\in T$ if $v\in V_{i-1}$.  The set of border edges is
denoted by $Bd(T, i)$.  Note that no edge in $Bd(T, i)$ is incident to two free vertices, otherwise the edge is in $E_{i-1}$ but not dominated by $T$, and then any border edge is incident to exactly
one free vertex.  We can see that an edge of $B_i$ incident to $y_i$ is not dominated by $T$ if and only if it is incident to a free vertex, and any edge in $T'\setminus T$ that is not incident to
$x_i$ is a border edge.  Then, for any border edge set $Z\subseteq Bd(T, i)$, $T\cup Z \in tr(E_i)$ only if each free vertex has a border edge $e\in Z$ incident to it.  Since any border edge is
incident to exactly one free vertex, for any $Z\subseteq Bd(T,i)$ such that $T\cup Z$ is irredundant and for any edge $vz_h\in Z$ with free vertex $z_h$, $P_{E_i}(e, T\cup Z)$ is always $\{vz_h\}$.
This implies that $T\cup Z$ is in $tr(E_i)$ only if $Z\subseteq Bd(T,i)$ includes exactly one edge incident to each free vertex.  We call such an edge set $Z$ a \emph{selection}.  We observe that all
border edges are dominated by $Z$.  We have the following lemma which is straightforward to prove.

\begin{lemma}\label{lem:l0} For any edge subset $Z$  with $Z\cap T=\emptyset$, there holds $T\cup Z \in tr(E_i)$ only if $Z$ is a selection.
\end{lemma}

An edge $e\in T$ is called \emph{redundant} if all edges in $P_{E_{i-1}}(e, T)$ are border edges and no edge $y_iz_h$ is in $P_{E_i}(e, T)$. 

\begin{lemma}\label{lem:l2} If $T$ has a redundant edge, then any selection $Z$ does not satisfy $T\cup Z\in tr(E_i)$.
\end{lemma}

\begin{proof} Let $e$ be a redundant edge of $T$. Since any border edge $f$ is incident  to a free vertex $z_h$, any selection $Z$ should contain one edge incident to $z_h$ and then if $f$ is incident
  to $e$, we have $f\notin P_{E_i}(e,T\cup Z)$. Since no edge $y_iz_h$ is in $P_{E_i}(e,T)$, there holds that $P_{E_i}(e,T\cup Z)=\emptyset$ for any selection $Z$.
\end{proof}

Let $X_T := \{ e\in Bd(T, i) \mid \exists e'\in T\ \textrm{and}\ P_{E_i}(e', T\cup \{ e\})\subseteq Bd(T, i)\}$.  The addition of any edge $e\in X_T$ to $T$ transforms an edge $e'$ of $T$ into a
redundant one with respect to $T\cup \{e\}$, and thus by Lemma \ref{lem:l2} for any $Z\subseteq Bd(T, i)$, $T\cup Z\in tr(E_i)$ holds only if $Z\cap X_T = \emptyset$. Therefore, the following follows. 

\begin{lemma}\label{lem:l3} If a free vertex is not incident to an edge in $Bd(T, i)\setminus X_T$, then any $Z\subseteq Bd(T, i)$ does not satisfy $T\cup Z\in tr(E_i)$.
\end{lemma}

One can hope that we can characterize the selections $Z$ not intersecting $X_T$ such that $T\cup Z\in tr(E_i)$ and be able to use it for listing the extra skip-children. Unfortunately, checking whether
there is such a selection $Z$ is \np\!\!-complete.

\begin{framed}
\noindent \textsf{\sc Irredundant Minimal Transversal (IMT)}.\\
\noindent {\bf Input.}\ A minimal transversal $T$ of $tr(E_{i-1})$. \\
\noindent {\bf Output.}\ Is there a selection $Z$, $Z\cap X_T=\emptyset$, such that $T\cup Z\in tr(E_i)$?
\end{framed}

\begin{theorem} \label{thm:IMT}
    IMT is NP-complete.
\end{theorem}

\begin{proof}
  We reduce the 3SAT problem to IMT. Let $\mathcal{C}=C_1,C_2,\cdots,C_m$ be the sets of clauses of a 3SAT instance over variables $x_1,...,x_n$. We create an instance of IMT as follows:
    \begin{itemize}
    \item $V_{i-1}:=\{ x_j \mid j\leq n \} \sqcup \{ \overline{x_j} \mid j\leq n \} \sqcup \{ C_h \mid h\leq m\} \sqcup \{ z_j \mid j\leq n \}\sqcup \{ y_j \mid j\leq n\} \sqcup \{ w \}$
    \item $E_{i-1}:= \{ x_jC_h \mid x_j\in C_h \} \sqcup \{ \overline{x_j}C_h \mid \overline{x_j}\in C_h \} \sqcup \{ x_jz_j\} \mid j\leq n \} \sqcup \{ \overline{x_j}y_j \mid j\le
    n \} \sqcup \{ x_j\overline{x_j} \mid j\leq n \}$
  \item $V_i:=V_{i-1} \sqcup \{ x,y \}$
  \item $E_i:= E_{i-1} \sqcup \{ yC_h \mid h\leq m \}\sqcup \{ xy, xw \}$
  \item $T:=\{ x_j\overline{x_j} \mid j\leq n \}\sqcup \{xw\}$
    \end{itemize}

\begin{figure}[ht]
\begin{center}
\includegraphics[scale=0.7]{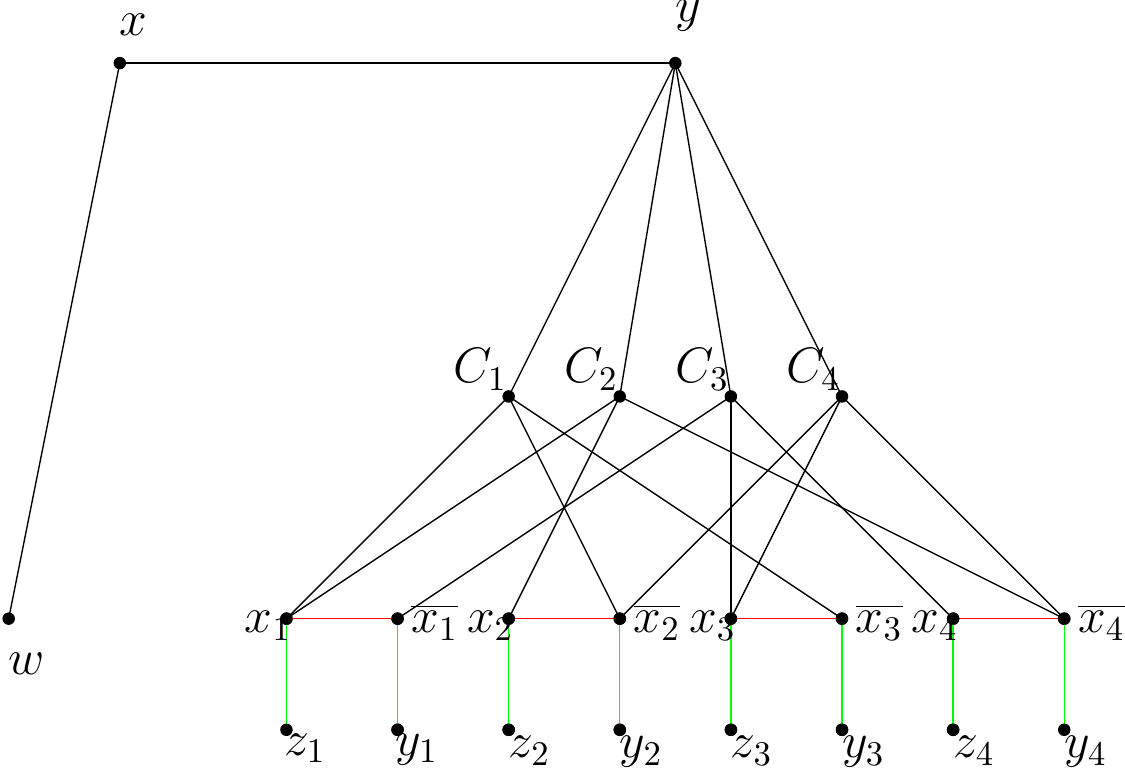}
\end{center}
\end{figure}

Notice that $X_T=\emptyset$, and then any selection $Z$ is such that $Z\cap X_T=\emptyset$. Now we claim that the 3SAT instance is satisfiable if and only if there is an edge set $Z\subset Bd(T,i)$
with $ Z\cap T=\emptyset$ such that $T\cup Z$ is included in $tr(E_i)$.

 Note that here we have $Bd(T, i) = \{ x_jC_h \mid x_j\in C_h \}\cup \{\overline{x_j}C_h \mid \overline{x_j}\in C_h \}$.  Notice now that a subset $Z$ of $Bd(T, i)$ is such that $T\cup Z\notin
 tr(E_i)$ if and only if $Z$ contains an edge of $\inc{x_i}$ and an edge of $\inc{\overline{x_i}}$ for some $i\leq n$.

Let $f:\{ x_1,\cdots,x_n \}\to \{ 0,1 \}$ be an assignment to the
 variables which satisfies the 3SAT formula.
Then consider the following subset $Z$ of $Bd(T, i)$, 
 \begin{align*}
   Z&:=\left(\bigcup\limits_{x_j\mid f(x_j)=1} \inc{x_j}\cap Bd(T, i)\right) \cup \left(\bigcup\limits_{x_j\mid f(x_j)=0} \inc{\overline{x_j}}\cap Bd(T, i)\right).
\end{align*}

Clearly, since $f$ satisfies the formula, $Z$ is a selection in $Bd(T, i)$ since otherwise a clause would not be satisfied by $f$.  Notice now that by construction, either $\inc{x_j}\cap Z=\emptyset$
or $\inc{\overline{x_j}}\cap Z=\emptyset$ for every $j\leq n$, and then there exists $Z'\subseteq Z$ such that $T\cup Z'$ is in $tr(E_i)$.

Assume now that there exists a selection $Z$ such that $T\cup Z\in tr(E_i)$.
Let $f:\{ x_1,\cdots,x_n \}\to \{ 0,1 \}$ be such that:

\begin{align*}
 f(x_j)&:=\begin{cases} 1 & \textrm{if $\inc{x_j}\cap Z\neq \emptyset$} \\ 0 & \textrm{if $\inc{\overline{x_j}}\cap Z\neq \emptyset$}. \end{cases}
\end{align*}

Notice first that $f$ is well-defined.  Indeed, assume that for some $x_j$, we have $\inc{x_j}\cap Z\neq \emptyset$ and $\inc{\overline{x_j}}\cap Z\neq \emptyset$.  Then the private neighbor of
the edge $x_j\overline{x_j}$ with respect to $T\cup Z$ would be empty, contradicting the fact that $T\cup Z\in tr(E_i)$.  Now since $Z$ is a selection of $Bd(T, i)$, for every $h\leq m$, $Z\cap
\inc{C_h}\neq \emptyset$ and then there exists either $x_j\in C_h$ with $f(x_j)=1$ or $\overline{x_j}\in C_h$ with $f(x_j)=0$.  Thus $f$ satisfies all clauses.  
\end{proof}

\begin{figure}[ht]
\begin{center}
\includegraphics[scale=0.45]{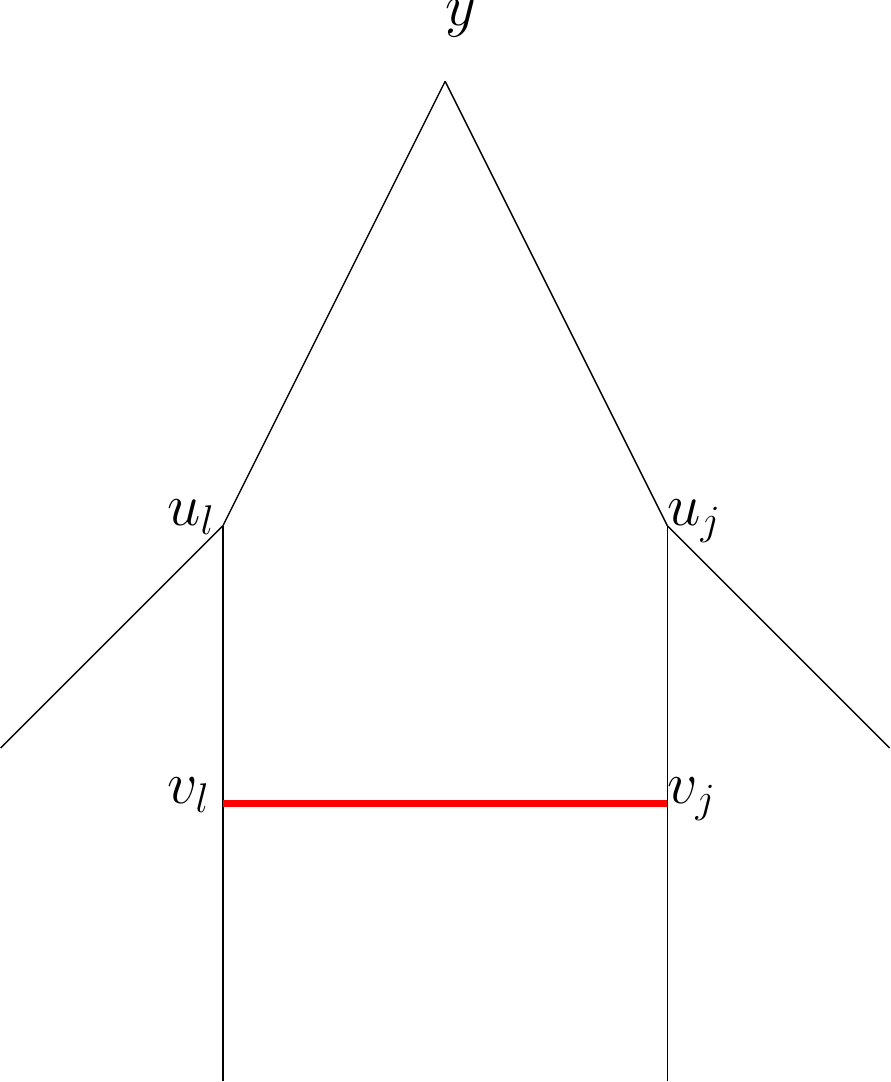}~~
\includegraphics[scale=0.45]{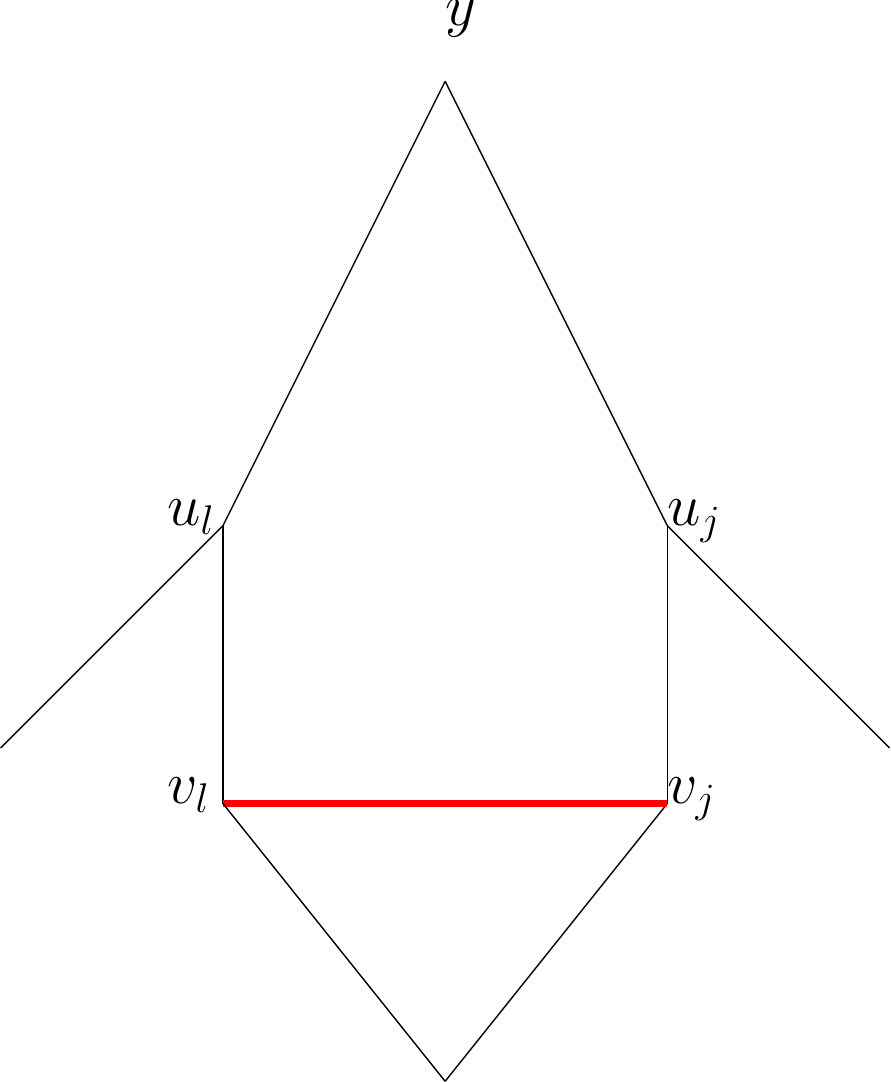}
\caption{Examples of  $H$-patterns: $u_l$ and $u_j$ are the free vertices and $v_lv_j$ is the edge in $T$}
 \label{fig:H}
\end{center}
\end{figure}

In order to overcome this difficulty, we identify a pattern, that we call an \emph{$H$-pattern}, that makes the problem difficult.


\begin{definition}[$H$-Pattern] A vertex set $\{ z_{\ell},v_{\ell},z_j,v_j\}$ is an \emph{$H$-pattern} if $z_{\ell}$ and $z_j$ are free vertices, $v_{\ell}v_j$ is in $T$, and $v_{\ell}v_j$ has two
  non-border private neighbors in $E_{i-1}\setminus T$: one is adjacent to $v_{\ell}$ and the other to $v_j$.  We also say that the edges $z_{\ell}v_{\ell}$, $z_jv_j$ and $v_{\ell}v_j$ induces an $H$-pattern.
\end{definition}

We will see that the difficulty of IMT comes from the presence of $H$-patterns.  Indeed, for an $H$-pattern $\{ z_{\ell},v_{\ell},z_j,v_j \}$, any private neighbor of $v_{\ell}v_j$ is adjacent to
either $z_{\ell}v_{\ell}$ or to $z_jv_j$, thus we cannot add both to a selection $Z$ since in that case $P_{E_i}(v_{\ell}v_j,T\cup Z)$ will be empty.  Let $H_T$ be the set of border edges included in
an $H$-pattern.  In the next two subsections we will see how to list selections including no edge from $H_T$, and those that do.

\begin{lemma}\label{lem:HyperReduit} If $T$ has no redundant edge, then $T\cup Z\in tr(E_i)$ holds for any selection $Z\subseteq Bd(T, i)\setminus (X_T\cup H_T)$.
\end{lemma}

\begin{proof} From the definition, $T\cup Z$ dominates all the edges in $E_i$ and for each $e\in Z$ it holds that $P_{E_i}(e,T\cup Z)\ne \emptyset$. Since $Z$ includes no edge from $H_T\cup X_T$, and
  $T$ has no redundant edge, one easily checks from Lemmas \ref{lem:l0}, \ref{lem:l2} and \ref{lem:l3} by case analysis that any edge $e\in T$ has a private neighbor $f$ that is adjacent to no border
  edge, or an edge $y_iz_h$ is adjacent to $e$ and not to edges in $T\setminus \{e\}$.  Thus, either $f\in P_{E_i}(e,T\cup Z)$ or $y_iz_h\in P_{E_i}(e,T\cup Z)$. These imply that $T\cup Z$ is in
  $tr(E_i)$.
\end{proof}

\subsection{Dealing with Redundancies} \label{ssec:redundancy}

The lemmas above demonstrate how to construct transversals $T'\in tr(E_i)$ from $T$, but some generated transversals may not be extra skip-children of $T$.  This is because such $T'$ can be also
generated from another transversals in $tr(E_{i-1})$.  Such redundancies happen for example when two edges $f_1$ and $f_2$ in $T'$ have private neighbors only in $B_i$, but after the removal of either
one from $T'$, the other will have a private neighbor outside $B_i$.  Assuming in this case that $f_1\in T$ and $f_2\in T'\setminus T$, it holds that $T'$ can be generated from $T$ or from $(T\setminus
\{f_1\})\cup \{f_2\}$. And since the number of selections $Z$ such that $T\cup Z\in tr(E_i)$ can be arbitrarily large, we need to avoid such redundancies.  

To address this issue, we state the following lemmas to characterize the edges not to be added to selections $Z$ such that $T\cup Z$ is an extra skip-child of $T$.  We say that a border edge $vz_\ell$
is \emph{preceding} if there is an edge $vz_h$ in $T$ satisfying $P_{E_{i-1}}(vz_h, T)\subseteq N[vz_\ell]$ and $y_iz_{\ell}<y_iz_h$, and denote the set of preceding edges by $X'_T$.  We also say that
an edge $vz_h\in T$ is {\em fail} if $P_{E_{i-1}}(vz_h, T)\subseteq Bd(T, i)$, $y_iz_h$ is in $P_{E_i}(vz_h, T)$, and no edge $wz_{\ell}\in P_{E_{i-1}}(vz_h, T)$ satisfies $y_iz_h < y_iz_\ell$.

\begin{lemma}\label{lem:l0-r} For any selection $Z$ including a preceding edge, $T\cup Z$ is not an extra skip-child of $T$.
\end{lemma}

\begin{proof} We can assume without loss of generality that $T\cup Z\in tr(E_i)$, otherwise the statement holds. Suppose that there are several edges $v_1z_{h_1},\cdots,v_1z_{h_p}$ in $T$ that are
  adjacent to preceding edges included in $Z$, and among them let $v_jz_{h_j}$ be such that $y_iz_{h_j}$ is greater than any $y_iz_{h_\ell}$ for $1\leq \ell \leq p$ and $\ell\ne j$, and let
  $v_jz_{\ell_j}$ be the preceding edge included in $Z$ such that $P_{E_{i-1}}(vz_{h_j},T)\subseteq N[v_jz_{\ell_j}]$. Let $t$ be the index of $y_iz_{h_j}$ in our ordering of the edges of $E$ and let
  $\cF_t$ be the first $t$ edges of $E$ in our ordering (which includes of course $y_iz_{h_j}$).  Since $P_{E_{i-1}}(vz_{h_j},T)\subseteq N[v_jz_{\ell_j}]$ no edge of $T$ but $v_jz_{h_j}$ is adjacent
  to $y_iz_{h_j}$ otherwise $P_{E_i}(v_jz_{h_j},T\cup Z)$ would be empty, and then $P_{E_i}(v_jz_{h_j},T\cup Z)=\{y_iz_{h_j}\}$. From the choice of $v_jz_{h_j}$ it follows that for any edge $e\in T$
  that is adjacent to a neighbor $z_h$ of $y_i$ and such that $y_iz_h>y_iz_{h_j}$ there exists an edge $f\in E_{i-1}\cap P_{E_i}(e,Q'_t(T\cup Z, |E_i|))$. Thus, every edge in $T$ has a private neighbor in $\cF_t$
  and hence $Q'_t(T\cup Z, |E_i|)$ includes all edges of $T$.  But since $y_iz_{h_j}$ has index $t$ and $P_{\cF_t}(v_jz_{h_j},T\cup Z)=P_{E_i}(v_jz_{h_j},T\cup Z)=\{y_iz_{h_j}\}$, we can conclude that
  $Q'_{t-1}(T\cup Z,|E_i|)$ does not contain $v_jz_{h_j}$, and then the skip-parent of $T\cup Z$ does not include $v_jz_{h_j}$. Therefore, $T\cup Z$ is not an extra skip-child of $T$.
\end{proof}

\begin{lemma}\label{lem:l2-r} If $T$ has a fail edge, then $T\cup Z$ is not an extra skip-child of $T$ for any selection $Z$.
\end{lemma}

\begin{proof} Let $vz_h$ be a fail edge of $T$ and let $Z$ be a selection. Suppose without loss of generality that $T\cup Z\in tr(E_i)$. Since $P_{E_{i-1}}(vz_h,T)\subseteq Bd(T,i)$ and each free
  vertex should be incident to an edge in $Z$ we can conclude that $P_{E_i}(vz_h,T\cup Z)=\{y_iz_h\}$. Now let $t$ be the index of $y_iz_h$ in the ordering of $E$ and let $\cF_t$ be the first $t$
  edges in this ordering. Assume also that $Q'_t(T\cup Z,|E_i|)$ contains all edges of $T$, otherwise $T\cup Z$ is not an extra skip-child of $T$. But, since $P_{\cF_t}(v_jz_{h_j},T\cup
  Z)=P_{E_i}(v_jz_{h_j},T\cup Z)=\{y_iz_h\}$, we can conclude that $Q'_{t-1}(T\cup Z,|E_i|)$ does not contain $v_jz_h$, and then the skip-parent of $T\cup Z$ does not include $v_jz_{h_j}$. Thus,
  $T\cup Z$ is not an extra skip-child of $T$.
\end{proof}

We are now able to characterize exactly those selections $Z$ not intersecting $H_T$ and such that $T\cup Z$ is an extra skip-child of $T$. 

\begin{lemma} \label{lem:good-selection} Suppose that $T$ has neither redundant edge nor fail edge and any free vertex is incident to an edge in $Bd(T,i)$. Then, $T\cup Z$ with $T\cap
  Z=\emptyset$ is an extra skip-child of $T$ including no edge of $H_T$ if and only if $Z$ is a selection including no edge of $X_T\cup X'_T\cup H_T$.
\end{lemma}

\begin{proof} The only if part is clear from lemmas \ref{lem:l0}, \ref{lem:HyperReduit}, \ref{lem:l0-r} and \ref{lem:l2-r}.  Let us now prove the if part.  Suppose that we have a selection $Z$
  including no edge in $X_T\cup X'_T\cup H_T$. By Lemmas \ref{lem:l3} and \ref{lem:HyperReduit} it holds that $T\cup Z\in tr(E_i)$.  Suppose now that $Q(T\cup Z, i)\ne T$.  Then, let us consider the
  computation of $Q(T\cup Z, i)$ in $E_i$: we compute $Q'(T\cup Z, |E_i|)$, $Q'(Q'(T\cup Z, |E_i|), |E_i|-1)$, and so on.  Let $F$ be the set of edges not in $B_i$ incident to the neighbors of $y_i$
  in $V_i$, \ie $F:= \inc{N(y_i)}\setminus \inc{y_i}$. First notice that $T\setminus F= Q(T\cup Z,i)\setminus F$. So, $T$ and $Q(T\cup Z,i)$ can differ only on edges in $F$. So, let $vz_h$ be the
  first edge removed among $T\setminus Q(T\cup Z, i)$ in this operation sequence, \ie, $Q'_{j+1}(T\cup Z, |E_i|)$ includes all edges in $T\setminus Q(T\cup Z, i)$, but $Q'_j(T\cup Z, |E_i|)$ does not
  include $vz_h$. Let $\cF_j$ be the first $j$ edges in our ordering.  Notice that all edges in $E_{i-1}$ are in $\cF_j$. Then, we can see that $P_{E_i}(vz_h, T\cup Z) = \{ y_iz_h\}$.  This implies
  that any private neighbor in $P_{E_{i-1}}(vz_h, T)$ is dominated by some edges in $Z$, and no edge in $(T\setminus \{vz_h\})\cup Z$ is adjacent to $z_h$.  We further see that if there is a private
  neighbor $uv\in P_{E_{i-1}}(vz_h, T)$ that is not a border edge, then no border edge is incident to $u$. Indeed, $u$ is not a free vertex and is necessarily in $V_{i-1}$ and if there is a border
  edge $uz_\ell$ this edge should be in $E_{i-1}$ and since it should be dominated by $T$ and $vu\in P_{E_{i-1}}(vz_h, T)$, there would exist an edge in $T$ incident to $z_\ell$ contradicting that
  $uz_\ell$ is a border edge. Similarly if there is a non border $uz_h\in P_{E_{i-1}}(vz_h, T)$, then no border edge is incident to $u$.

  Suppose that all edges in $P_{E_{i-1}}(vz_h, T)$ are border edges.  Since $vz_h$ is not a fail edge, there exists an edge $wz_\ell\in P_{E_{i-1}}(vz_h, T)$ satisfying $y_iz_h < y_iz_\ell$.  This
  implies that the edge $wz_\ell\in P_{\cF_j}(vz_h, Q'_{j+1}(T\cup Z, |E_i|))$, and then $vz_h$ should be in $Q'_j(T\cup Z,|E_i|)$, otherwise $wz_\ell$ would not be dominated by $Q'_j(T\cup
  Z,|E_i|)$, thus yielding a contradiction.

  Suppose now that there is a non-border edge $vu$ in $P_{E_{i-1}}(vz_h, T)$.  We note that $u$ can be $z_h$ so that $vu = vz_h$.  Since $vz_h$ is not in $Q'_j(T\cup Z, |E_i|)$, there should be a
  border edge in $Z$ adjacent to $vz_h$.  We can observe that $u$ is incident to no edge in $T\setminus \{vz_h\}$, thus any border edge adjacent to $uv$ is incident to $v$. If
  $P_{E_{i-1}}(vz_h,T)\subseteq \inc{v}$, then since $T$ has no preceding edge any border edge $vz_\ell\in Z$ satisfies that $y_iz_h < y_iz_\ell$. This implies that $Q'_{j+1}(T\cup Z, |E_i|)$ includes
  no such border edge, and $uv$ is a private neighbor of $vz_h$ in $P_{\cF_j}(vz_h, Q'_{j+1}(T\cup Z, |E_i|))$.  This implies that $vz_h$ is included in $Q'_j(T\cup Z, |E_i|)$, yielding a
  contradiction. So, there is a non border edge $z_hw\in P_{E_{i-1}}(vz_h,T)$. Let $z_s$ and $z_p$ be free vertices adjacent respectively to $z_h$ and $v$ and such that $z_hz_s$ and $vz_p$ are in
  $Z$. If $z_s\ne z_p$, then $\{z_s,z_h,z_p,v\}$ would form an $H$. So there is at most one free vertex $z_s$ such that $vz_s$ and $z_hz_s$ are in $Z$. If such a $z_s$ exists, then one of $z_hz_s$ and
  $vz_s$ is not in $Z$. And then in this case either $wz_h$ or $vu$ is in $P_{\cF_j}(vz_h, Q'_{j+1}(T\cup Z, |E_i|))$, contradicting that $vz_h$ is not in $Q'_j(T\cup Z,|E_i|)$.  If $z_h$ is not adjacent
  to a border edge, then $wz_h\in P_{\cF_j}(vz_h,Q'_{j+1}(T\cup Z,|E_i|))$, and then again $vz_h$ would be in $Q'_j(T\cup Z,|E_i|)$.

  From the discussion, we have that $T\setminus Q(T\cup Z, i) = \emptyset$.  Since $Q(T\cup Z, i)$ and $T$ are both minimal in $tr(E_{i-1})$, we have $Q(T\cup Z, i)=T$.
\end{proof}

As a corollary we have the following.

\begin{prop}\label{prop:enum-skip1} One can enumerate with polynomial delay and space all the extra skip-children of $T$ that do not contain edges of $H_T$. 
\end{prop}

\begin{proof} If $T$ has redundant edges or fail edges or has a free vertex not incident to an edge in $Bd(T,i)\setminus X_T$, then by Lemmas \ref{lem:l2}, \ref{lem:l3} and \ref{lem:l2-r} we can
  conclude that $T$ has no extra skip-child. Since we can compute $X_T$ in polynomial time and check in polynomial time whether an edge is redundant or is a fail edge, this step can be done in
  polynomial time. So, assume $T$ has no redundant edge nor fail edges and every free vertex is incident to an edge in $Bd(T,i)\setminus X_T$. By Lemma \ref{lem:good-selection} by removing all edges
  in $H_T\cup X_{T}\cup X'_T$, any selection $Z$ is such that $T\cup Z$ is a skip-child of $T$. One easily checks that the enumeration of these selections can be reduced to the enumeration of the
  minimal transversals of a hypergraph of degree at most $2$, and in these hypergraphs minimal transversals can be enumerated with polynomial delay and polynomial space \cite{EGM03}.
\end{proof}

\subsection{Dealing with the Presence of $H$-Patterns}\label{sec:5}

As we saw in Theorem \ref{thm:IMT}, it is hard to enumerate all extra skip-children having some edges in $H$-patterns from a given transversal $T\in tr(E_{i-1})$. Let us call these children {\em
  $H$-children}.  We approach this difficulty by introducing a new parent-child relation among $H$-children, and enumerate them by traversing the forest induced by the new relation.  In this way, we
now do not follow the skip-parent skip-child relation for $H$-children.  However, the root of each tree in the induced forest is a transversal obtained with the skip-child skip-parent relation. Let us
be more precise now.  For two sets $S$ and $S'$ of edges we write $S<_{lex} S'$ if $\min (S\Delta S') \in S$, called lexicographical ordering. 

Hereafter, we consider an extra skip-child $T'=T\cup Z$ of $T\in tr(E_{i-1})$ such that $T'\cap H_T \ne \emptyset$.  Let $H^*(T'):=\{v_hz_h, v_{\ell}z_{\ell}, v_hv_\ell\}$ be the lexicographically
minimum $H$-pattern among all $H$-patterns of $T$ that includes an edge of $Z$.  Without loss of generality, we assume that $v_{\ell}z_{\ell}$ is in $Z$.  Let $uz_h$ be the edge in $Z$ incident to
$z_h$. Notice that such an edge exists because $z_h$ is a free vertex. Then, we define the {\em slide-parent} $Q^*(T', i)$ of $T'$ by $T'\cup \{ v_hz_h \} \setminus \{ uz_h, v_hv_\ell\}$.

\begin{lemma}\label{lem:slide-uniq} The slide-parent of $T'$ is well-defined and is a member of $tr(E_i)$.
\end{lemma}

\begin{proof} Since $z_h$ is a free vertex for $T$, $Z$ includes exactly one edge in $Z$, thus $uz_h$ is uniquely determined, and thus the slide-parent is uniquely defined.  Since $uz_h$ is a border
  edge, either $u\not\in V_{i-1}$ or $u$ is incident to an edge of $T$.  This together with that $v_hz_h$ and $v_{\ell}z_\ell$ dominate all edges in $N[v_hv_\ell]$ leads that $Q^*(T', i)$ dominates
  all edges in $E_i$.

  By the addition of $v_hz_h$ to $T'$, no edge in $T'\setminus \{ uz_h, v_hv_\ell\}$ loses its private neighbor.  The edge $v_hz_h$ is adjacent to no edge in $T'\setminus \{ uz_h, v_hv_\ell\}$, and
  then $v_hz_h\in P_{E_i}(v_hz_h,Q^*(T',i)$.  These imply that $Q^*(T', i)$ is a member of $tr(E_i)$.
\end{proof}

The slide-parent of $T$ has less edges than $T$, thus the (slide-parent)-(slide-child) relationship is acyclic, and for each $T'\in tr(E_i)$, there is an ancestor $T''\in tr(E_i)$ in the (slide-parent)-(slide-child)
relation such that the skip-parent of $T''$ has no $H$-pattern.  Similar to the depth-first search versions of Berge's algorithm \cite{KavvadiasS05,MurakamiU14}, we will traverse the (slide-parent)-(slide-child)
relation to enumerate all transversals including $H$-pattern edges. The following follows from the definition of slide-parent. 

\begin{proposition}\label{prop:slide-child} Any slide-child $T'$ of $T''$ is obtained from $T''$ by adding two edges and remove one edge.
\end{proposition}

The computation of the slide-parent of any $T'\in tr(E_i)$ including edges of $H$-patterns can be easily done in polynomial time: compute its skip-parent $T$ in polynomial time, choose $H^*(T)$ and
then compute its slide-parent in polynomial time as described above.  Proposition \ref{prop:slide-child} shows that there are at most $m^3$ candidates for slide-children, thus the enumeration of
slide-children can be done with polynomial delay and polynomial space.

\begin{lemma} For any $T'\in tr(E_i)$, all its slide-children can be enumerated  with polynomial delay and polynomial space.
\end{lemma}

\section{Summary} \label{sec:6}

We can now summarize the steps of the algorithm. 

\begin{itemize}
\item all transversals in $tr(E_1)$ can be enumerated with polynomial delay and 
 polynomial space, since they include at most two edges from $N[b_1]$.
\item In Section \ref{sec:4} (second paragraph), we have explained how to enumerate all non-extra skip-children
with  polynomial delay and polynomial space
\item By Proposition \ref{prop:enum-skip1} all the extra skip-children not including any edges of $H$-patterns can be enumerated with polynomial delay and polynomial space. 

\item Let $\cT$ be the set of all extra skip-children including some edges from $H$-patterns of their skip-parent.  We proved in Section \ref{sec:5} how to enumerate all the members of $\cT$ with
  polynomial delay and polynomial space, by traversing the (slide-parent)-(slide-child) relation, \ie, recursively listing all the slide-children of the current visiting transversal.

\item Therefore, by executing these three enumeration algorithms for each
 transversal in $T\in tr(E_{i-1})$, we can generate all the members in 
 $tr(E_i)$ with polynomial delay and polynomial space.
\end{itemize}

All these show that the conditions of Proposition \ref{prop:strategy} are satisfied. And thus we can state our main result.

\begin{theorem}\label{thm:main}
All edge minimal dominating sets in a graph $G$ can be enumerated with
 polynomial delay and polynomial space.
\end{theorem}




\section{Conclusion}\label{sec:7}

In this paper, we propose a polynomial delay polynomial space algorithm for listing all minimal edge dominating sets in a given graph.  This improves drastically the previously known algorithms which
were incremental output-polynomial and use exponential space. We state furthermore that usual approaches with Berge's algorithm involves an NP-complete problem, and thus it is difficult with usual
approaches of Berge's algorithm to produce an efficient algorithm.  To cope with this difficulty, we introduce a new idea of ``changing the traversal routes in the area of difficult solutions'' (the
notion of skip-children and the removal of edges involved in $H$-patterns).  Based on this idea, we give a new traversal route on these difficult solutions, that is totally independent from Berge's
traversal route (the (slide-parent)-(slide-child) relation).  As a result, we are able to construct a polynomial delay polynomial space algorithm. 

The idea of changing the traversal routes seems to be new and to be able to apply to many other kind of algorithms in enumeration area.  Interesting future works are applications of this idea to other kind
of enumeration algorithms, \eg the one used by Lawler et al. for enumerating maximal subsets \cite{lawler80} or other algorithms for enumerating minimal transversals (see for instance \cite{EGM03}).

\bibliographystyle{plain}
\bibliography{bib}

\end{document}